\documentclass{article}
\usepackage[T1]{fontenc}
\usepackage{iclr2027_conference,times}
\usepackage{amsmath,amssymb,amsthm,mathtools}
\usepackage{graphicx,booktabs,array,multirow,float}
\usepackage{wrapfig,pifont}
\usepackage[section]{placeins}
\usepackage{microtype,xcolor,url,hyperref}
\usepackage{tikz}
\usetikzlibrary{positioning,arrows.meta,fit,backgrounds,calc}

\definecolor{dsblue}{HTML}{2A78D6}
\definecolor{dsorange}{HTML}{EB6834}
\definecolor{dsgray}{HTML}{8A8984}

\hypersetup{colorlinks=true,linkcolor=black,citecolor=black,urlcolor=black,
  pdftitle={DeepInfer-Mac: Running a 750B-Class MoE on a 24 GB Laptop},
  pdfauthor={Haorui Li}}
\title{DeepInfer-Mac:\\Running a 750B-Class MoE on a 24\,GB Laptop}
\author{%
\textbf{Haorui Li}\\
Southeast University\\
\href{mailto:haoruileee@gmail.com}{\texttt{haoruileee@gmail.com}}
}
\iclrfinalcopy 
\makeatletter
\renewcommand{\@maketitle}{\vbox{\hsize\textwidth\centering
  \def\@makefnmark{\hbox{\@textsuperscript{\normalfont\@thefnmark}}}%
  {\LARGE\scshape\@title\par}
  \def\And{\end{tabular}\hspace{2em}%
    \begin{tabular}[t]{c}\rule{\z@}{24pt}\ignorespaces}%
  \begin{tabular}[t]{c}\rule{\z@}{24pt}\@author\end{tabular}\par
  \vskip 0.6em
  {\normalfont\small\url{https://github.com/haoruilee/DeepInfer-Mac}\par}
  \vskip 0.3in minus 0.1in
}}
\makeatother
\newcommand{\method}{\textsc{DeepInfer-Mac}}
\newcommand{\relu}[1]{\left[#1\right]_{+}}

\newcommand{\bfround}{\operatorname{bf16}}
\newcommand{\q}{\mathbf q}
\newcommand{\kv}{\mathbf k}
\newcommand{\RR}{\mathbb R}
\newtheorem{proposition}{Proposition}
\newtheorem{lemma}{Lemma}
\newcommand{\GenerationSpeedRange}{0.21--0.72}
\newcommand{\GenerationPeakFootprint}{16.73}
\newcommand{\GenerationCache}{4.495}
\newcommand{\LongPromptTokens}{3{,}764}
\newcommand{\LongPromptTTFT}{309.0}
\newcommand{\SwapStart}{6.54}
\newcommand{\SwapEnd}{13.36}
\newcommand{\MatchedTokens}{108}
\newcommand{\MatchedBlocks}{9}
\newcommand{\MatchedCache}{4.495}
\newcommand{\MatchedFinding}{Adding next-layer speculation changes aggregate throughput from 0.92 to 1.11 tokens/s, a 21.0\% increase relative to overlap alone.}
\newcommand{\SpeculationGain}{21.0}
\newcommand{\SpeculationExtraPayload}{5.4}
\newcommand{\RetrievalOutcome}{Both retrieval requests return the planted access code exactly.}
\newcommand{\GenerationRows}{English & 22 & 69 & 65.0 & 0.21 & 12.88 \\
Python & 27 & 25 & 26.6 & 0.25 & 12.94 \\
Chinese & 25 & 57 & 28.2 & 0.72 & 12.94 \\
Retrieval 1K & 971 & 2 & 145.1 & -- & 14.64 \\
Retrieval 4K & 3764 & 2 & 309.0 & -- & 16.73 \\}
\newcommand{\MatchedRows}{Serial & 0.77 & 1.00 & 2.02 \\
Overlap resident experts & 0.92 & 0.97 & 2.01 \\
Add next-layer speculation & 1.11 & 0.80 & 2.12 \\}

\begin{document}
\maketitle
\lhead{} 
\renewcommand{\headrulewidth}{0pt}
\raggedbottom
\begin{abstract}
DeepInfer-Mac runs the 750B-class DeepSeek-V4.1-Flash model on a MacBook Air with 24\,GB of unified memory. A mixed low-bit conversion occupies 338\,GiB on the internal SSD. About 7.6\,GiB of weights remain resident; routed experts and Engram rows are loaded when needed. This separation makes the model executable, but it leaves each generated token waiting for a succession of storage reads. We organize expert weights so that reads fill reusable GPU buffers directly, and couple cache allocation with the lifetime of those buffers. In end-to-end runs, the system generates English, Chinese and code at \GenerationSpeedRange{} tokens per second, with a largest sampled process footprint of \GenerationPeakFootprint{}\,GiB across the tested requests. Matched decode comparisons show a \SpeculationGain{}\% throughput gain from next-layer prefetching over resident-expert overlap alone, at the cost of additional reads. Longer prompts remain costly. We describe the design, provide the implementation and measurement scripts, and identify the limits of this operating point: a model far larger than memory can run locally, while storage and prefill still govern the user's wait.
\end{abstract}

\section{Introduction}
A 750B-class model appears far beyond a laptop with 24\,GB of memory. Even after compression, the checkpoint used in this work occupies 338\,GiB. Yet a mixture-of-experts model does not use all of those parameters for each token. DeepSeek-V4.1-Flash selects six of the 384 routed experts in each of its 40 layers \citep{deepseek41flash}. Most of the model can stay on disk while a much smaller working set moves through memory.

This paper follows that observation through to a running system. \method{} executes the model on a single MacBook Air, using its internal SSD to hold routed experts and Engram tables. Attention, routers and other frequently used weights remain in unified memory. A selected expert is read directly into a reusable GPU buffer, computed, and retained until its space is needed. The model's size determines the storage requirement; the active experts and their reuse determine the cost of generating a token.

Making the working set fit is only the beginning. A GPU cluster can keep many tokens in flight and overlap their communication with computation \citep{deepinfer,deepseekv3}. On this laptop, one token passes through one layer at a time. The router must run before the system knows which experts to read, and those reads must finish before the missing experts can be computed. The next layer then repeats the same sequence. The SSD has enough capacity for the model, but its bandwidth and this dependency chain govern the pace of generation.

Our design spends memory to shorten that chain. The dense weights stay resident, expert reads avoid an extra user-space copy, and repeated expert use is served from a bounded cache. The same budget also has to cover buffers whose GPU work has not finished. We therefore couple cache allocation, expert batching and buffer lifetime. This makes the memory constraint enforceable even when optional reads are still in flight.

The resulting system generates text on the target laptop. Runs with \GenerationCache{}\,GiB of reserved expert capacity produce English, Chinese and Python at \GenerationSpeedRange{} tokens per second. A \LongPromptTokens{}-token retrieval prompt completes in \LongPromptTTFT{} seconds to the first token. These measurements use full prefill and local weights throughout. They also show why feasibility needs to be reported alongside latency: fitting the working set does not make a long prompt cheap.

Once the system fits, reading the next layer's likely experts can hide some of the remaining wait. A controlled comparison restores the same model state, warmup and continuation before each schedule. Adding demand-first speculation increases aggregate decode throughput by \SpeculationGain{}\% over resident-expert overlap in these matched blocks, while reading \SpeculationExtraPayload{}\% more expert data. The improvement comes from moving useful transfers earlier, with enough capacity reserved for the next demand group. The system exposes all three schedules so that the tradeoff can be tested on other requests and machines.

The contribution is a complete implementation and a measured account of this operating point: a 750B-class MoE, one 24\,GB laptop, and SSD-backed execution. We explain which data must remain resident, how the rest can be supplied safely, and what limits the resulting performance. Operator-level comparisons with the reference and a study of indexer arithmetic support the implementation. The implementation, converted weight layout and raw measurements are available for others to reproduce and extend this case study.

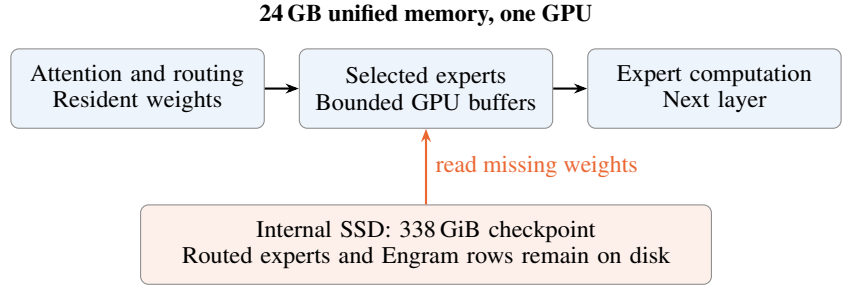
\begin{figure}[htbp]
\centering
\begin{tikzpicture}[
 font=\small,
 box/.style={draw=black!45, rounded corners=3pt, align=center, minimum height=1.05cm, inner sep=5pt},
 arr/.style={-{Stealth[length=5pt]}, line width=0.8pt},
 node distance=0.45cm]
 \node[box, fill=dsblue!8, text width=3.0cm] (dense) {Attention and routing\\Resident weights};
 \node[box, fill=dsblue!8, text width=3.0cm, right=of dense] (cache) {Selected experts\\Bounded GPU buffers};
 \node[box, fill=dsblue!8, text width=3.0cm, right=of cache] (compute) {Expert computation\\Next layer};
 \node[box, fill=dsorange!10, text width=7.2cm, below=1.0cm of cache] (disk) {Internal SSD: 338\,GiB checkpoint\\Routed experts and Engram rows remain on disk};
 \draw[arr] (dense) -- (cache);
 \draw[arr] (cache) -- (compute);
 \draw[arr, dsorange] (disk.north) -- node[right, font=\footnotesize] {read missing weights} (cache.south);
 \node[above=0.18cm of cache, font=\small\bfseries] {24\,GB unified memory, one GPU};
\end{tikzpicture}
\caption{The full checkpoint lives on the SSD. Memory holds the frequently used weights and a changing set of experts. The router determines which expert records enter the buffers; the GPU consumes them in place. Engram lookup follows a separate row cache.}
\label{fig:overview}
\end{figure}

\section{From model size to working set}
\label{sec:background}
The distinction between stored parameters and active parameters is what makes this experiment possible. The checkpoint contains 15,360 routed experts, but a decoding token selects 240 of them across the model. Their weights amount to several GiB, and a cache can avoid reading some of them again. Attention, routing, shared experts and the output head require about 7.6\,GiB of resident weights. This leaves room for expert reuse and temporary computation within the laptop's physical memory.

There are two other substantial parts of the model. Compressed sparse attention stores a compressed history and uses an indexer to choose the positions a query attends to. Engram looks up hashed $n$-grams in large embedding tables. Its tables occupy 114.4\,GiB after conversion, but a lookup reads only the rows addressed by the current tokens. Both mechanisms separate a large stored state from the portion needed immediately. Table~\ref{tab:model} summarizes the model and machine; Appendix~\ref{app:format} gives the full storage layout.

\begin{table}[htbp]
\centering\small
\caption{The model, checkpoint and target machine. Storage sizes are binary GiB; 24\,GB is the machine's marketed memory configuration. The parameter count excludes quantization metadata (Appendix~\ref{app:parameter-count}).}
\label{tab:model}
\begin{tabular}{@{}lp{0.64\linewidth}@{}}
\toprule
Model & DeepSeek-V4.1-Flash, 748.5B text weights, 40 layers \\
Routed computation & 6 of 384 experts per layer, plus one shared expert \\
Converted checkpoint & 337.6\,GiB; mixed 2/3/4-bit routed experts and 4-bit Engram rows \\
Resident weights & 7.6\,GiB; mainly attention, shared experts and routing \\
Machine & MacBook Air M4, 10-core GPU, 24\,GB unified memory \\
Storage & 1\,TB internal SSD \\
Workload & One sequence at a time; inference runs locally \\
\bottomrule
\end{tabular}
\end{table}

Prior work has established that sparse activation can make flash storage useful for inference \citep{llmflash,powerinfer,moeoffload}. The question here is how far that separation carries a much larger MoE on a small unified-memory machine. A server sends token activations to devices that already hold the experts. This machine has only one GPU, so it brings the experts to the token. That reverses the direction of data movement and makes the storage path central to the design.

The resource tradeoff is direct. A larger expert cache avoids more storage reads, but takes memory from activations, sequence state and the operating system. More simultaneous reads may keep the SSD busy, but they also keep more buffers unavailable for reuse. The runtime must manage these demands together. Otherwise a nominally small cache can still leave too little room to execute the next layer.

\section{Keeping the token moving}
\label{sec:design}
The design follows one token from routing to expert computation. Its aim is to make each required byte available with little additional storage or copying, then retain that byte when reuse is likely.

\subsection{A direct path from SSD to GPU}
\label{sec:fetch}
Each routed expert is stored as one contiguous record, already arranged in the quantized layout consumed by MLX \citep{mlx}. A cache miss assigns the record a slot in preallocated GPU arrays. The read writes directly into that slot, and the GPU later uses the same memory. Unified memory removes the need for an additional user-space staging copy. The expert files use uncached reads to avoid retaining a second copy in the operating system's page cache.

The router selects the experts before the host issues missing reads. Experts that are already resident can be computed while the others arrive. Once their buffers are ready, the missing experts pass through the same gathered matrix operations. This preserves the model's selected experts while changing when their work is issued. The on-disk arrangement and execution layout therefore serve the same purpose: an expert moves from storage to computation without being unpacked into another full-sized weight buffer.

\subsection{Memory that can actually be reused}
\label{sec:replication}
A cache entry is useful only if keeping it does not prevent the next request from running. This matters under lazy GPU execution: a host call may have returned while the GPU still refers to its weights. Replacing those weights too early can change a computation already in flight.

We keep one global expert-buffer budget and account for the different record sizes used by different layers. The available slots determine how many experts a prefill batch may contain. Demand acquisition protects both hits and newly loaded experts together, and that protection remains until the relevant GPU work is safe. Optional reads are admitted only while enough slots remain for the next demand group. Appendix~\ref{sec:cache-budget} states the capacity invariant and its assumptions.

Within those constraints, replacement is least recently used. A hit saves an entire expert read, so cache capacity has a direct effect on the token's storage demand. The budget bounds expert buffers rather than the whole process: dense weights, attention state, temporary arrays and the operating system still need memory. We measure the process separately instead of treating a cache setting as a memory measurement.

\subsection{The limits of reading ahead}
\label{sec:chain}
Routing creates a dependency that storage bandwidth alone cannot remove. A layer runs attention and its router, reads any missing experts, and computes their outputs before the next layer can proceed. A useful first approximation is
\begin{equation}
 T_{\mathrm{token}} \simeq T_{\mathrm{compute}}
       + \frac{D_{\mathrm{miss}}}{B_{\mathrm{SSD}}}
       - T_{\mathrm{overlap}},
\label{eq:chain}
\end{equation}
where $D_{\mathrm{miss}}$ is the number of expert bytes actually fetched and $T_{\mathrm{overlap}}$ is the transfer time hidden by independent computation. The expression identifies two different opportunities: reuse reduces the bytes required, while scheduling may hide part of the time needed to deliver them.

Computing cached experts supplies some independent work, but a miss-heavy layer offers little of it. Predicting a later layer's experts gives the SSD a longer head start. We use the established next-layer router predictor and demand-first ordering of \citet{moeoffload}: guesses are issued after the current layer's demand reads have completed. Speculative entries receive low insertion priority, an established way to limit cache pollution \citep{qureshi2007}. The budget and lifetime rules apply to these reads as well.

A prediction can still be costly after those safeguards. An unnecessary read uses bandwidth, and a read that is already in flight can outlast the idle interval that motivated it. Demand-first admission guarantees neither lower latency nor freedom from contention with a later layer. We therefore evaluate speculation against a schedule that simply overlaps resident expert computation with missing reads. Section~\ref{sec:parallel-eval} compares them on identical work.

\section{Executing the complete model}
\label{sec:impl}
Attention state, Engram lookup and the residual mixers share the machine with the expert cache. \method{} implements the reference model's compressed attention, indexers, Engram lookup and Hyper-Connections in MLX, with a common forward path for prefill and decode. A streaming converter reads the released tensors, converts them one at a time and writes the final checkpoint without first storing a second complete copy of the release.

The dense FP8 tensors retain their values when converted to MLX's MXFP8 representation. Routed experts use affine quantization at two to four bits per weight, selected by layer; Engram tables use four-bit rows. The conversion is consequently lossy as a whole. The per-layer precision choices trade reconstruction error against storage and read volume. Appendix~\ref{app:format} records the formats and allocation sizes.

Long prompts introduce a different working set. Many tokens pass through a layer together, so prefill can touch most of its experts. We process those experts in batches that fit the same cache budget. The engine also supports bounded replay, inherited from the upstream execution strategy, which omits much of the later-layer work on older prompt tokens. Because truncating this computation is an approximation, the end-to-end measurements use full prefill. Prefix reuse and the optional replay mode remain available for applications that make different latency tradeoffs.

Correctness checks follow the boundaries of the implementation. Small configurations run against the official PyTorch reference to test attention, routing, state updates and the vision encoder. Other tests force buffer eviction races, unfinished reads and repeated state restoration. These checks establish behavior for the exercised cases; they do not make the quantized checkpoint equivalent to the release. English, Chinese, code and retrieval requests supply complementary evidence that the complete text path runs with real weights. Additional interfaces and the speculative decoder are available in the repository, but are not needed for the central result.

\section{What the laptop can do}
\label{sec:eval}
Once the model fits, latency becomes the visible constraint. We begin with complete requests, then isolate the effect of scheduling on their decode steps. The end-to-end and matched runs use the M4 laptop in Table~\ref{tab:model}, local Engram tables, full prefill and the same converted checkpoint. Other desktop applications remain open; this is a shared-machine measurement rather than an isolated hardware benchmark.

\subsection{Generation within the memory budget}
\label{sec:e2e}
Table~\ref{tab:results} reports requests in English, Chinese and Python, followed by retrieval from longer prompts. Each request starts with an empty expert cache and a fresh sequence state. We count decode throughput from the intervals after the first emitted token, so a short answer does not turn its prefill work into an apparent decoding speedup.

\begin{table}[htbp]
\centering\small
\caption{End-to-end requests with \GenerationCache{}\,GiB of reserved expert buffers. TTFT is time to the first token after model construction. Footprint is the largest process value sampled at 0.5-second intervals during the request. Retrieval outputs are too short for a useful throughput comparison.}
\label{tab:results}
\begin{tabular}{@{}lrrrrr@{}}
\toprule
Request & \shortstack{Prompt\\tokens} & \shortstack{Output\\tokens} & \shortstack{TTFT\\(s)} & \shortstack{Decode\\(tok/s)} & \shortstack{Footprint\\(GiB)} \\
\midrule
\GenerationRows
\bottomrule
\end{tabular}
\end{table}

The three short requests generate \GenerationSpeedRange{} tokens per second. Their outputs include a physical explanation, a Python function and a Chinese explanation; the complete token sequences and text are preserved in the logs. These examples demonstrate operation across several kinds of input. Their timing also reflects a changing machine state, so the rows are not a comparison of language or task speed. \RetrievalOutcome{} The larger retrieval request has \LongPromptTokens{} input tokens and takes \LongPromptTTFT{} seconds to produce its first token.

The working set remains much smaller than the checkpoint. Across these requests, the largest sampled process footprint is \GenerationPeakFootprint{}\,GiB, while the checkpoint occupies 337.6\,GiB on disk. Figure~\ref{fig:working-set} shows how footprint changes during prefill and generation. Footprint, expert-buffer capacity and MLX allocations measure different things: the dense weights plus the configured cache are not a sufficient estimate of the peak.

\begin{figure}[htbp]
\centering
\includegraphics[width=\linewidth]{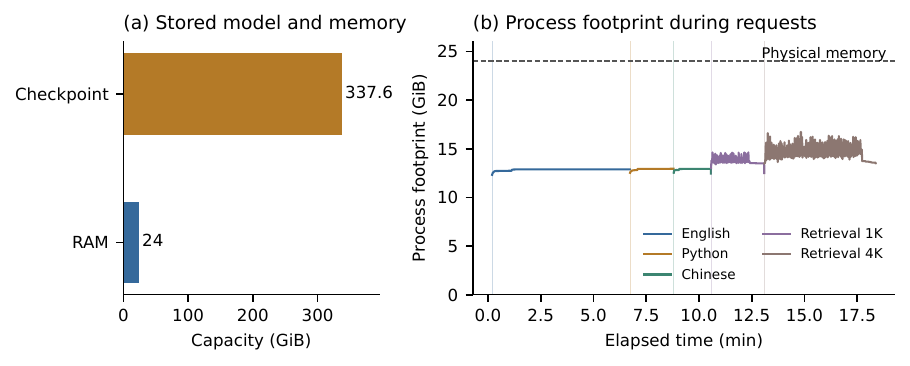}
\caption{Checkpoint storage and measured working memory. Left: the stored checkpoint is much larger than physical memory. Right: process footprint during the requests, with request boundaries marked. The footprint includes more than the expert cache and is sampled rather than an exact peak measurement. Host swap activity is reported separately.}
\label{fig:working-set}
\end{figure}

Host swap was active before model loading and grew during the suite. Its change cannot be assigned entirely to the inference process on this shared desktop. The experiment establishes execution on a 24\,GB machine with a measured process footprint; it does not establish swap-free operation. Appendix~\ref{app:measurement} gives the memory accounting and host counters.

Long prompts expose the other part of the tradeoff. During decode, a token selects only a few experts in each layer. A prompt contains many routing decisions, so prefill can visit most experts and must also build its attention state. A small resident working set therefore coexists with substantial startup latency. The full-prefill retrieval requests make that cost visible instead of hiding it behind the optional bounded-replay approximation.

\FloatBarrier
\subsection{Trading storage traffic for latency}
\label{sec:parallel-eval}
The dependency chain in Section~\ref{sec:chain} leaves intervals in which storage and computation can overlap. We compare three schedules: serial execution, computing resident experts while missing experts load, and adding demand-first next-layer speculation. Each measured block starts from the same prompt snapshot, the same serial warmup and the same expert-cache state. Its input continuation is fixed, and the generated token ids must agree. Configuration order rotates across rounds.

\begin{table}[htbp]
\centering\small
\caption{Matched-work schedule comparison. Each schedule processes \MatchedTokens{} measured tokens across \MatchedBlocks{} blocks, using \MatchedCache{}\,GiB of reserved expert capacity. Throughput is total measured tokens divided by their elapsed time. The ratio is the median per-block elapsed time relative to the corresponding serial block; lower is faster.}
\label{tab:matched}
\begin{tabular}{@{}lrrr@{}}
\toprule
Schedule & Tokens/s & \shortstack{Relative\\time} & \shortstack{Expert payload\\(GiB/token)} \\
\midrule
\MatchedRows
\bottomrule
\end{tabular}
\end{table}

\MatchedFinding{} Expert payload grows by \SpeculationExtraPayload{}\%, and all compared schedules return the same measured output tokens. Figure~\ref{fig:matched} shows the incremental effect of speculation in each paired block. It is faster than overlap alone in all nine blocks, though the gain varies between prompts. These warmed blocks isolate scheduling; their rates are distinct from the complete-request measurements above. The result supports the measured schedules in this configuration; it does not rank them across SSDs, cache sizes or concurrent workloads.

\begin{figure}[htbp]
\centering
\includegraphics[width=\linewidth]{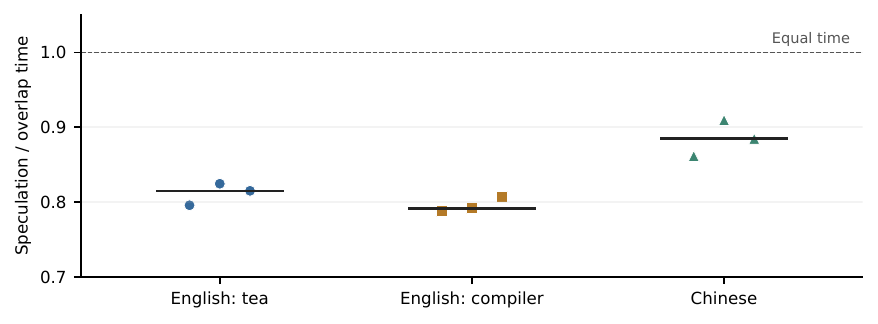}
\caption{The incremental effect of next-layer speculation, relative to resident-expert overlap alone. Each point is one prompt--round pair with the same continuation and initial cache state; short horizontal bars mark medians. The dashed line at one marks equal time. Rotation balances schedule position, but does not eliminate thermal or shared-machine effects.}
\label{fig:matched}
\end{figure}

This is the central scheduling tradeoff on the laptop. A cache hit eliminates a necessary transfer. A speculative read instead places a bet that the transfer will be useful soon enough to repay its bandwidth and residency cost. Keeping the SSD occupied is therefore an incomplete objective. The useful quantity is the time removed from the token's critical path after accounting for that extra work.

\subsection{What the correctness evidence covers}
\label{sec:quality}
The implementation passes 78 tests with the official reference available. The suite covers operator comparisons on small configurations, cache and buffer lifetime, sampling, state restoration and serving behavior. The matched runs add a check that scheduling changes preserve the measured tokens on the full checkpoint. Neither check establishes the quality loss introduced by converting the original weights.

An older teacher-forced evaluation compared the converted model with API continuations. The API vectors and the checkpoint came from different release dates, so their disagreement mixes model revision and compression effects. Appendix~\ref{app:api} preserves the per-case results. A same-revision, broader quality evaluation remains necessary before claiming that the low-bit checkpoint preserves the original model's capabilities.

\section{When computation becomes visible}
\label{sec:indexer}
Expert reads dominate much of decoding, while long prompts also expose the cost of choosing attention keys. The indexer scores many query--key pairs before selecting a small set. Its inputs have already been quantized to FP4 values with power-of-two scales, which creates an opportunity to use faster arithmetic without introducing another quantization step.

For the model's 128-element dot products, we give a sufficient condition under which BF16 multiplication with FP32 accumulation produces the same final BF16 value as the reference FP32 computation. The condition bounds the range of the block scales and excludes underflow and overflow. The discrete inputs then make every partial sum exactly representable in FP32. Appendix~\ref{app:exact} gives the proof and the boundary of the guarantee; the runtime does not yet enforce it with a guard and fallback for every input.

On captured inputs from a 32K-token prompt, the BF16 path matched 1.078 billion scores. The matrix multiplication was 1.5--1.9 times faster in the measured blocks, and fusing the subsequent head reduction removed repeated passes over the intermediate scores. These are local improvements, whose end-to-end value depends on how much time a request spends in the indexer.

Avoiding some scores altogether was less useful in our exploratory study. DoSA \citep{dosa} uses bounds to reject keys before full scoring, but the model's existing candidate selection leaves a smaller domain to prune. Our adaptation saved only a modest fraction of arithmetic on the captured queries, and its heuristic BF16 allowance is not a sound general bound. Screening therefore remains disabled. The detailed measurements in Appendix~\ref{app:indexer} explain this decision without making it part of the throughput claim.

\section{Related work}
\label{sec:related}
Running a model larger than memory has a substantial history. FlexGen schedules weights, activations and KV state across GPU, CPU and disk to serve large batches \citep{flexgen}. LLM in a Flash exploits activation sparsity and storage-friendly layouts to load only needed weights \citep{llmflash}. PowerInfer places frequent neurons on the GPU and computes others on the CPU \citep{powerinfer}. These systems establish the broader strategy of trading storage or host computation for scarce accelerator memory. Our contribution is its realization and measurement for this 750B-class model on a 24\,GB unified-memory laptop.

MoE offloading makes the working set explicit through expert routing. \citet{moeoffload} combine an expert cache with next-layer prediction and demand-first speculative loading. MoE-Infinity and ProMoE develop expert caching and prefetching, including coordination between prediction and inference \citep{moeinfinity,promoe}. Pre-gated MoE changes routing to expose later expert choices earlier \citep{pregated}, while Fiddler uses CPU expert computation to reduce transfers \citep{fiddler}. We adopt established caching and prediction ideas; the single-sequence SSD setting changes their costs. Our budget and lifetime protocol couples their allocations so that outstanding reads and GPU work leave enough capacity for the next demand group.

Several open-source projects already bring DeepSeek-V4-generation models to Apple silicon. ds4 provides a C and Metal engine and the API-sampled vectors used in our earlier checks \citep{ds4}. flash-moe streams experts from SSD, and MLX ports target V4-Flash and V4.1-Flash \citep{flashmoe,ssdmoe,pipenetwork}. \method{} builds on the official reference structure and MLX, with local Engram lookup, a streaming converter and a common sequence-state path. This paper provides an auditable account of its resource use and scheduling behavior at 24\,GB. A common cross-system benchmark remains future work.

The indexer study addresses a separate cost within sparse attention. GVR, LiteTopK and DoSA reduce selection or scoring work while preserving specified selection behavior \citep{gvr,litetopk,dosa}. Accurate GEMM on lower-precision matrix units is also an established line of work, including the Ozaki scheme \citep{mukunoki2025}. Our measured BF16 path uses the discrete structure already present in the model inputs, and its equality proof has explicit conditions. The screening study is an exploratory application of existing ideas, with the numerical limitations documented in the appendix.

\section{Discussion and conclusion}
\label{sec:conclusion}
A small laptop can execute a model whose stored parameters are far larger than its memory when the active working set is small enough and the remaining weights can be supplied from storage. \method{} realizes that separation for DeepSeek-V4.1-Flash. The implementation runs the full text path, and the measurements report the memory and latency required to do so.

The laptop makes the model locally accessible at the cost of storage traffic and a substantial wait for longer prompts. More concurrency does not remove the dependency between routing and expert availability; it changes which transfers can happen before they are needed. Memory spent retaining useful experts can avoid reads. Memory and bandwidth spent on wrong predictions can instead delay the next layer. The runtime must account for both effects within the same budget.

The evidence is limited to one laptop, one converted checkpoint and a single sequence at a time. The desktop is shared, host swap is active, and the generated examples are not a broad quality benchmark. The low-bit conversion and optional bounded replay also need separate quality evaluation; the latter is disabled in these runs. The code and logs allow these limits to be tested on other machines and workloads. Measurements on a second storage device and a same-revision quality reference would help separate hardware sensitivity from compression loss.

\label{lastmainpage}
\section*{AI use statement}
Generative AI tools assisted with code, experiments and manuscript drafting. The accompanying logs and tests document the checks performed. The author takes responsibility for the final manuscript and implementation.

\section*{Reproducibility statement}
An archived snapshot of the implementation, converter, regression tests and measurement logs accompanies this preprint as ancillary material. Development continues at \url{https://github.com/haoruilee/DeepInfer-Mac}. Appendix~\ref{app:measurement} gives the commands, measurement definitions and source fingerprints. Tables and plots are generated directly from the saved measurements. The weight layout and numerical arguments are in the appendix. Reference tests require the official source through \texttt{DSV41\_REF}; without it those tests are skipped. Historical logs remain available but are not used to estimate the fresh schedules' speedups.

\clearpage
\bibliography{references}
\bibliographystyle{iclr2027_conference}
\clearpage
\appendix
\section{Weight layout}
\label{app:format}
\subsection{Parameter accounting}
\label{app:parameter-count}
Unpacking the source tensor shapes gives 543.582 billion routed-expert weights, 196.614 billion Engram entries and 8.299 billion other text weights, for a total of 748.495 billion. The backbone excluding Engram is 551.881 billion weights, consistent with the model card's rounded 552B backbone count \citep{deepseek41flash}. Quantization scales describe the stored representation and are excluded from these counts. The optional draft network adds 14.225 billion weights; the vision encoder and aligner add 0.485 billion. Neither is needed in the text-generation runs.

An earlier 754.6B count adds 6.144 billion Engram scale entries to the text weights. We use ``750B-class'' to avoid treating that metadata as additional model parameters. The audit is reproducible with \path{scripts/parameter_audit.py} and saved in \path{logs/parameter_audit.json}. The 337.6\,GiB storage total below includes the optional components and quantization metadata actually stored on disk.

\subsection{Stored representation}
Each routed expert is one record of three projections: gate $W_1$ and up $W_3$ ($2304\times5120$) and down $W_2$ ($5120\times2304$). Each projection is stored as packed integers, per-group scales and per-group biases in MLX's affine format with group size 64, and the record is padded to a multiple of 16\,KiB. Table~\ref{tab:layers} gives the bit widths per layer. We chose them with \path{scripts/quant_damage.py}, which runs real layers once with the converted experts and once with the original FP4 experts and compares hidden states layer by layer. The first four layers, the layers that own KV caches or index keys, and the decoder band received 3 bits, the last layer 4 bits, and the rest 2-bit gate/up projections. The DSpark draft experts use 2-bit gate/up and 3-bit down projections.

\begin{table}[htbp]
\centering
\small
\caption{Expert bit widths ($W_1$/$W_3$/$W_2$) and record sizes.}
\label{tab:layers}
\begin{tabular}{@{}llr@{}}
\toprule
Layers & Bits & Record \\
\midrule
0--3, 8, 14, 20--38 & 3/3/3 & 15.5\,MB \\
4--7, 9--13, 15--19 & 2/2/3 & 12.5\,MB \\
39 & 4/4/4 & 19.9\,MB \\
\bottomrule
\end{tabular}
\end{table}

\begin{table}[htbp]
\centering
\small
\caption{Converted weights of DeepSeek-V4.1-Flash as stored on disk. Sizes are in GiB.}
\label{tab:format}
\begin{tabular}{@{}>{\raggedright\arraybackslash}p{0.40\linewidth}>{\raggedright\arraybackslash}p{0.28\linewidth}rl@{}}
\toprule
Component & Format & Size & Location \\
\midrule
Dense layers: attention, shared experts, routers, mHC, Engram projections & \texttt{mxfp8}, lossless & 6.9 & memory \\
Output head and other global tensors & affine 8-bit / BF16 & 0.7 & memory \\
Embedding & BF16, memory-mapped & 1.2 & SSD \\
Routed experts, $40\times384$ & affine 2/3/4-bit, $g=64$ & 208.3 & SSD \\
DSpark layers: dense parts / experts ($3\times128$) & \texttt{mxfp8} / affine 2/3-bit & 0.7 / 4.5 & SSD \\
Vision encoder & BF16 & 0.9 & SSD \\
Engram tables, $2\times384$M rows & affine 4-bit, $g=32$ & 114.4 & SSD \\
\midrule
Total on disk & & 337.6 & \\
\bottomrule
\end{tabular}
\end{table}

\clearpage
\section{Per-case API agreement}
\label{app:api}
\begin{table}[htbp]
\centering
\small
\caption{Historical teacher-forced checks against an older API snapshot (September 10 versus the October 17 checkpoint). NLL is the mean negative log-likelihood per continuation token. The macro-average gives the twelve one-token retrieval cases most of the weight; all-token agreement is 123/141 (87.2\%), and non-retrieval agreement is 111/129 (86.0\%). These are not same-revision quantization measurements.}
\label{tab:api}
\begin{tabular}{@{}lrrrrr@{}}
\toprule
Case & Prompt & Continuation & NLL (ours) & NLL (API) & Top-1 agreement \\
\midrule
Prose & 15 & 32 & 0.733 & 0.038 & 78.1\% \\
Code & 17 & 32 & 0.637 & 0.278 & 84.4\% \\
Arithmetic & 18 & 1 & 0.000 & 0.000 & 100\% \\
Prose & 21 & 32 & 0.244 & 0.115 & 90.6\% \\
Italian & 15 & 32 & 0.512 & 0.354 & 90.6\% \\
Retrieval ($\times12$) & 127--16{,}385 & 1 & $\le$0.006 & 0.000 & 100\% each \\
\midrule
Mean over 17 cases & & & 0.127 & 0.046 & 96.7\% \\
\bottomrule
\end{tabular}
\end{table}

\section{Capacity and buffer lifetime}
\label{sec:cache-budget}
The store couples its allocation to the maximum live working set. The plan accounts for
heterogeneous record sizes, sets the prefill batch size, and reserves slots for demand.
These guarantees concern capacity and buffer lifetime; next-layer prediction and demand-first
prefetching follow prior work \citep{moeoffload}.

Let layout $g$ have aligned record size $c_g$, serve $L_g$ layers, and receive $n_g$ slots.
Let $D$ be the largest decode routing set, including the draft model, and let $M$ be the
requested slab budget. We require $M\geq M_{\min}=4D\sum_g c_g$ and reject smaller budgets
before allocating. For $R=M-M_{\min}$, the initial plan is
\begin{equation}
n_g=4D+\left\lfloor\frac{R L_g}{\sum_h L_h c_h}\right\rfloor,
\qquad
B_g=\max\!\left(D,\min\!\left(16,\left\lfloor n_g/4\right\rfloor\right)\right).
\end{equation}
Rounding leftovers buy additional slots only when their full record size fits. Thus
$\sum_g n_gc_g\leq M$ and $4B_g\leq n_g$. Prefill uses batches of at most $B_g$ experts;
decode requests already fit because $B_g\geq D$. The budget covers record capacity and bounds
the array payloads, not dense weights, KV caches, allocator overhead or process RSS.

\paragraph{Admission and leases.} Two previous request groups remain pinned while lazy GPU
work may reference them. An optional read is admitted only if (i) the number of unfinished,
unleased reads remains at most $n_g-3B_g$, and (ii) at least $B_g$ slots remain free or
immediately evictable. Condition (i) is needed even at cold start, when the pin sets are empty.
Demand acquisition starts reads and pins existing hits in the same critical section. The
outer decode lease covers both the resident and missing expert groups; each inner acquisition
transfers protection to the lazy pin sets before releasing its lease. A failed split read is
not considered finished until all of its writers stop.

\begin{proposition}
Under the single-compute-thread protocol, with requests of at most $B_g$ distinct experts and
at most two previous pinned groups, admitted optional reads cannot cause demand allocation to
run out of slots. Every acquired slot remains protected until its lease or lazy-compute pin
is released.
\end{proposition}
\begin{proof}
At a demand boundary, at most $2B_g$ slots are pinned and at most $n_g-3B_g$ contain unfinished
unleased reads. At least $B_g$ slots are therefore free or evictable. Protecting resident
members of the incoming request consumes no more than one slot per member and reduces the
number of additional slots needed by the same amount. The remaining misses fit. Allocation
and lease publication use one lock, and eviction excludes both leases and pins. Nested
leases are reference-counted, so releasing a subgroup cannot release the outer request.
\end{proof}
This is a capacity and lifetime guarantee. It does not bound I/O latency, cancel an SSD read,
or assert that a wrong prediction never displaces an old, unpinned cache entry.

\paragraph{Allocation audit.} The former allocator imposed a minimum of 96 slots separately
on every layout, after proportionally dividing the requested budget. This exceeded that
budget on the production manifest. Table~\ref{tab:budget-audit} reports a deterministic
calculation from that manifest, saved in \path{logs/cache_budget_audit.json}; these are not
RSS measurements. Regression tests exercise all experts at the minimum capacity, hold
speculative reads unfinished across successive demand batches, force the former hit-eviction
race, and inject both whole-record and split-read failures.

\begin{table}[htbp]
\centering\small
\caption{Aligned expert-record capacity in GiB. The revised allocation obeys the requested
budget; the historical timings elsewhere used the former allocator.}
\label{tab:budget-audit}
\begin{tabular}{@{}rrr@{}}
\toprule
Requested & Former allocator & Revised allocator \\
\midrule
2.0 & 4.285 & 1.993 \\
2.5 & 4.458 & 2.493 \\
3.5 & 5.159 & 3.500 \\
4.5 & 6.118 & 4.495 \\
5.0 & 6.610 & 4.995 \\
5.5 & 7.091 & 5.499 \\
6.5 & 8.061 & 6.494 \\
\bottomrule
\end{tabular}
\end{table}

\section{Indexer arithmetic and screening}
\label{app:indexer}
Long prompts take minutes on a laptop, and part of that time goes to the indexer. For a query with indexer heads $\q_1,\dots,\q_H\in\RR^{d}$ and signed head weights $w_1,\dots,w_H$, the reference scores each visible compressed key $\kv_j$ as
\begin{equation}
s_j \;=\; \bfround\!\Big(\sum_{h=1}^{H} \bfround\big(\relu{\bfround(\q_h^\top\kv_j)}\, w_h\big)\Big),
\label{eq:score}
\end{equation}
with the outer sum accumulated in FP32, and keeps the best $K$ keys. Before candidate filtering, scoring all visible query--key pairs grows quadratically with the prompt length. We studied two ways to make it cheaper: faster scoring, with equality under the conditions below, and screening keys before fully scoring them.

\subsection{Faster scoring with the same results}
\paragraph{BF16 products.} The reference computes each product $\q_h^\top\kv_j$ in FP32 and then rounds it to BF16. Both inputs are FP4 values with power-of-two scales, which BF16 represents exactly, and the following lemma shows when a BF16 matrix multiplication returns the same rounded result.
\begin{lemma}
\label{lem:exact}
Let $\q,\kv\in\RR^{128}$ hold E2M1 values with one power-of-two scale per block of 32 elements, each exactly representable as a finite BF16 value. Let $\sigma_b$ be the product of the two scales of block $b$. Assume that $\max_b\sigma_b/\min_b\sigma_b\le 2^{9}$, that $\tfrac14\min_b\sigma_b\ge 2^{-126}$, and that no partial sum overflows FP32. Then every partial sum of $\q^\top\kv$, in any order, is exactly representable in FP32, and a BF16 matrix multiplication that accumulates in FP32 and rounds once returns $\bfround(\q^\top\kv)$, the reference result.
\end{lemma}
Appendix~\ref{app:exact} gives the proof. MLX's BF16 matrix kernels accumulate in FP32, although we checked this empirically rather than for every kernel variant. Where the scale condition fails, both computations are ordinary FP32 sums whose last bit can depend on the order of summation, just as it can between different GPU kernels running the reference. On indexer inputs captured from a real 32K-token prompt, the BF16 version reproduced every one of about a billion scores.

\paragraph{Fused head reduction.} After the matrix multiplication, Equation~\eqref{eq:score} still needs a ReLU, a weighting and a rounding for every element, followed by a sum over heads. Written as array operations, this reads a large intermediate tensor three times. We replaced it with a single Metal kernel in which each thread handles one pair of query and key and walks over the heads, preserving the per-head roundings. On every input we tested it matches the array version. Both sum the heads in FP32, but a different summation order need not be identical for all possible inputs. Figure~\ref{fig:indexer}a reports the measured block timings and equality checks.

\paragraph{Fixed tie order.} A partition-based Top-$K$ leaves ties at the boundary to the kernel. We sort by a 64-bit key built from the BF16 score and the position, so that ties go to the earlier position; NaN scores are mapped to negative infinity. Selection thereby becomes a pure function of the scores, which an exact screening method must also reproduce.

\subsection{Skipping keys: an exploratory screening study}
\label{sec:dosa}
DoSA \citep{dosa} skips keys while keeping the selection exact. A seed set of keys gives a lower bound $\tau$ on the $K$-th best score, a cheap upper bound $U_j$ is computed for every other key, and keys with $U_j<\tau$ are dropped before full scoring. DoSA reports large speed-ups on GLM-5.2 and DeepSeek-V3.2 at contexts of half a million tokens and more. V4.1-Flash differs in three ways that could matter: it has fewer indexer heads, a smaller $K$, and its own two-level candidate selection. We therefore measured how well such a bound works on this model.

\paragraph{Setup.} We captured the indexer inputs of a real 32K-token prompt and studied every index layer except the one that chooses the candidate blocks, which needs all block maxima and cannot skip keys. For each layer we used the last decode query and sixteen queries from the prompt. A key is \emph{eligible} for a query if the query can see it and, in the later layers, if it lies inside the query's candidate blocks. As in DoSA, the bound evaluates the positive heads $P$ and the strongest negative heads $E$ exactly, and merges the remaining negative heads into a single term:
\begin{equation}
U_j = \sum_{h\in P} w_h\relu{\q_h^\top\kv_j} - \sum_{h\in E} |w_h|\relu{\q_h^\top\kv_j} - \relu{\textstyle\sum_{h\in N\setminus E}|w_h|\,\q_h^\top\kv_j}.
\end{equation}
This is a real-arithmetic bound because $\sum_i \relu{x_i}\ge\relu{\sum_i x_i}$. The historical script added heuristic BF16 slack and found no violation on its captured queries. That slack is not a certificate for Equation~\eqref{eq:score}: three positive heads with dot product 64.75 and weight 0.31640625 produce an exact rounded score of 62, above the heuristic upper value 61.9421. This implementation does not have a general exactness guarantee. The following costs describe the captured data only, and screening remains disabled. Its cost counts the fully scored seed keys, the dot products of the bound for every other eligible key, and the remaining heads for every key that survives.

\paragraph{Results.} Figure~\ref{fig:indexer}b and Table~\ref{tab:dosa} show the cost at the best head budget for each layer. A useful bound needs about two thirds of the heads, and even then screening saves only a tenth to a fifth of the indexer's work, before the overhead of compacting the survivors on the GPU. DoSA's joint bound for omitted positive heads did not drop a single additional key on this model, and a low-rank bound built from the leading singular vectors of the query heads failed as well, because those singular values decay slowly. Two properties of the model explain the difference from DoSA's setting. At these context lengths, $K$ is a large fraction of the keys, so the threshold $\tau$ sits low in the score distribution. And the model's own candidate selection already removes many keys before the indexer score matters. Arithmetic savings alone do not establish an end-to-end benefit once bounds, compaction and selection are included. Screening remains disabled in the runtime. The picture may change at contexts of 128K tokens and more, but prefilling such a prompt takes hours on this machine.

\begin{figure}[H]
\centering
\includegraphics[width=\linewidth]{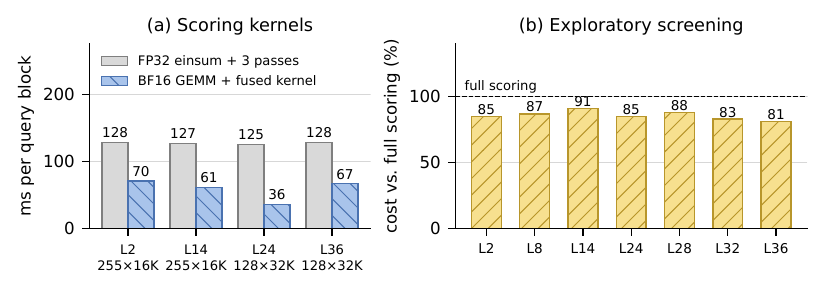}
\caption{(a) Indexer time per prefill query block on real inputs, before and after the BF16 matmul and the fused head reduction. Labels give the layer and the block shape as rows times keys. Scores and selected keys are identical on these blocks. The matmul alone runs 1.5 to 1.9 times faster and the head reduction 5 to 6 times faster. (b) Historical cost of the heuristic screening bound relative to fully scoring the eligible keys, per index layer, at the best head budget. The heuristic was checked on captured inputs but is not generally sound under BF16 rounding; the cost includes scoring the seed keys and is averaged over 17 queries of a 32K-token prompt with $K=512$.}
\label{fig:indexer}
\end{figure}

\subsection{Screening measurements}
\label{app:dosa}
Table~\ref{tab:dosa} lists, for each index layer, the best configuration of the bound averaged over 17 queries (the last decode query and 16 prefill queries), with up to 4096 of the most recent eligible keys as the seed set. The BF16 allowance is heuristic, so this table describes observed arithmetic cost rather than certified pruning. $e$ is the number of negative heads evaluated exactly, $m$ the resulting number of dot products per key (positive heads, $e$ negative heads and the merged term; it varies slightly between queries), and the cost includes scoring the seed keys. \path{scripts/dosa_study.py} reproduces the table and checks that no dropped key belongs to the true Top-$K$.

\begin{table}[H]
\centering
\small
\caption{Best screening configuration per layer (32K-token prompt; layers 2--14 see 16K compressed keys, layers 24--36 see 32K keys, of which about 16K are eligible). Survive is the fraction of the non-seed eligible keys whose bound is not below $\tau$.}
\label{tab:dosa}
\begin{tabular}{@{}lrrrr@{}}
\toprule
Layer & $e$ & $m$ (mean) & Survive & Cost vs.\ full scoring \\
\midrule
2 & 0 & 24 & 18.1\% & 85\% \\
8 & 0 & 23 & 31.7\% & 87\% \\
14 & 4 & 26 & 24.4\% & 91\% \\
24 & 4 & 21 & 38.5\% & 85\% \\
28 & 0 & 22 & 40.0\% & 88\% \\
32 & 0 & 24 & 8.9\% & 83\% \\
36 & 0 & 22 & 19.9\% & 81\% \\
\bottomrule
\end{tabular}
\end{table}

\FloatBarrier
\begin{samepage}
\subsection{Exact accumulation for FP4 inputs}
\label{app:exact}
E2M1 magnitudes are $\{0, 0.5, 1, 1.5, 2, 3, 4, 6\}$, i.e., integers in $\{0,\dots,12\}$ times $2^{-1}$. In block $b$ with scales $s_{q,b}$ and $s_{k,b}$, each product $q_ik_i$ is therefore an integer of magnitude at most $144$ times $\tfrac14\sigma_b$, where $\sigma_b=s_{q,b}s_{k,b}$ is a power of two. Let $\sigma_{\min}$ and $\sigma_{\max}$ be the smallest and largest block products. Every product, and hence every partial sum over any subset of the 128 terms, is an integer multiple of $\tfrac14\sigma_{\min}$. Its magnitude is at most $128\cdot144\cdot\tfrac14\sigma_{\max}<2^{15}\cdot\tfrac14\sigma_{\max}$. Expressed in units of $\tfrac14\sigma_{\min}$, a partial sum is an integer below $2^{15}\,\sigma_{\max}/\sigma_{\min}\le2^{24}$, which FP32's 24-bit significand represents exactly. Hence every FP32 addition in any order is exact, the accumulated value equals the real dot product, and the single final rounding to BF16 equals $\bfround(\q^\top\kv)$. This assumes that the block products stay in the normal FP32 range, i.e.\ $\tfrac14\sigma_{\min}\ge2^{-126}$ and no overflow. The quantizer clamps each scale at $2^{-126}$, so the assumption can fail for blocks of near-zero activations; on our captured inputs the scale exponents lie roughly between $-8$ and $4$. \hfill$\square$

\end{samepage}
\subsection{Top-\texorpdfstring{$K$}{K} order key}
\label{app:order}
For a BF16 score with bit pattern $b$ (16 bits), the key $\kappa(b) = 2^{16}-1-b$ if the sign bit is set and $b+2^{15}$ otherwise is monotone in the score value, after mapping $-0$ to $+0$. For position $j<2^{31}$ we select the $K$ largest values of $\kappa(b)\cdot2^{32} + (2^{31}-1-j)$, a 64-bit integer, with a partition. This yields the order ``higher score first, then earlier position'' regardless of the partition kernel. In the candidate-block level the query's own block has score $+\infty$. After NaNs are replaced by $-\infty$, this block ranks ahead of every finite competitor.

\section{Measurement protocol and evidence}
\label{app:measurement}
\subsection{Fresh end-to-end requests}
The generation suite uses \path{scripts/paper_benchmark.py} with a requested 4.5\,GiB cache, full prefill, greedy sampling and a limit of 128 output tokens for each of three short requests. The two retrieval requests have a 16-token output limit. The expert slab is cleared before each request; the host page cache is not purged. Prompt construction and tokenization precede the request timer. Model construction is timed separately, while work evaluated lazily on the first request is included in its latency.

Time to first token is the interval from calling the stream to receiving its first output. Decode throughput is $(n-1)/(t_n-t_1)$ for $n$ emitted tokens, including the terminal token if emitted. It therefore excludes prefill and does not count a final forward pass whose output is never used. The log retains every output id, complete text and intertoken interval. Retrieval throughput is omitted from the main table because its answers are too short to characterize sustained decoding.

A background thread samples macOS process resident size and physical footprint every 0.5 seconds through \texttt{proc\_pid\_rusage}. The log also records the process RSS high-water mark and MLX's peak allocation counter. The latter is reset before each request and includes allocations managed by MLX, not all process memory. Sampled maxima may miss short-lived peaks. Physical footprint and RSS are different accounting measures and are not added together.

Host-wide swap usage, VM counters and thermal status are recorded before and after each request. Swap usage rises from \SwapStart{}\,GiB before model loading to \SwapEnd{}\,GiB after the suite. Swap-in and swap-out counter deltas count pages transferred across the entire machine, not bytes attributable to the model. Other applications remain open, so these measurements describe execution in a shared desktop environment. No claim of zero swap or isolated-device throughput follows from them.

\subsection{Matched schedules}
The schedule suite uses three prompts, three rounds, eight serial warmup tokens and twelve measured tokens per prompt--round block. Every configuration restores the same prompt snapshot, clears the expert cache, and replays that warmup before consuming the same continuation. The starting cache fingerprint and generated ids must match. Schedule position rotates across rounds. The workload contains 108 measured tokens per prompt across all schedules and 324 measured tokens in total.

Expert-read volume counts the record payload, excluding alignment padding, for both required and speculative reads submitted in the measured region. A final submission barrier attributes outstanding speculative submissions to that region; completion of optional work outside the timed forward calls is not charged as token latency. Cache reset drains outstanding reads before the next configuration. Thus the throughput measures forward-call latency, and the byte counter also exposes speculative work that may continue beyond it.

\subsection{Commands and provenance}
The two suites are run sequentially from the repository root, using one model process at a time:
\begin{quote}\small
\begin{verbatim}
python scripts/paper_benchmark.py --suite generation \
  --cache-gb 4.5 --tokens 128 --replay full \
  --out paper/logs/laptop_generation.json

python scripts/paper_benchmark.py --suite matched \
  --cache-gb 4.5 --tokens 12 --rounds 3 --warmup 8 \
  --replay full --out paper/logs/laptop_matched.json
\end{verbatim}
\end{quote}
The logs contain the operating-system and MLX versions, model configuration, metadata hashes, actual expert-buffer capacity and a SHA-256 fingerprint of the runtime and benchmark sources. The matched protocol also saves its per-block data in \path{laptop_matched.matched.json}. Figure and table generation reads these files directly. Prompts, token ids and source fingerprints allow changes in workload or implementation to be distinguished from changes in the machine's state.

\subsection{Historical evidence}
Earlier measurements used an allocator whose per-layout floor exceeded the requested global capacity. A run labeled 5.5\,GiB actually reserved 7.091\,GiB of expert records and observed 1.23 tokens/s. Its schedule comparisons advanced a shared generation, so configurations processed different tokens and inherited cache effects. Those records remain in the repository, but their reported 11--27\% improvements are not used as controlled estimates in this paper. Appendix~\ref{sec:cache-budget} preserves the allocation audit.

The historical API comparison used continuations sampled on September 10, 2026 and the converted October 17 checkpoint. Its per-case table is retained in Appendix~\ref{app:api}; the version mismatch prevents isolating quantization loss. Indexer measurements in \path{bench_kernels.log} and \path{dosa_study_v2.log} are separate captured-input experiments, not fresh end-to-end benchmarks.

Runtime regression validation with the official reference passes 78 tests. Eight sampling and converter cases fail against the earlier \texttt{a4df33c} package, and the former allocation overshoot and hit-eviction race are reproducible. Their logs are stored as \path{runtime_regression_baseline.log} and \path{runtime_regression_fixed.log}. These checks support the revised implementation; they are not throughput measurements.

\end{document}